\documentclass{article}

\usepackage{booktabs} 

\usepackage{amsthm}
\usepackage{amsmath}
\usepackage{amsfonts}
\usepackage{bbm}
\usepackage{graphicx}
\usepackage{subcaption}
\usepackage{url}
\usepackage{color}
\usepackage{fullpage}
\usepackage{float}

\newtheorem{theorem}{Theorem}[section]

\newtheorem{lemma}[theorem]{Lemma}
\newtheorem{assumption}[theorem]{Assumption}
\newtheorem{observation}[theorem]{Observation}

\title{Access to Population-Level Signaling as a Source of Inequality}
\author{Nicole Immorlica\thanks{Microsoft Research NE and NYC. Email: nicimm@gmail.com} \and Katrina Ligett\thanks{Hebrew University of Jerusalem. Email: katrina@cs.huji.ac.il. This work was funded in part by the HUJI Cyber Security Research Center in conjunction with the Israel National Cyber Directorate (INCD) in the Prime Minister's Office, the Israeli Science Foundation, and by a DARPA Brandeis subcontract.} \and Juba Ziani\thanks{California Institute of Technology. Email: jziani@caltech.edu. Ziani's research was supported in part by NSF grants CNS-1331343 and CNS-1518941, the Linde Graduate Fellowship at Caltech, and the inaugural PIMCO Graduate Fellowship at Caltech.}}

\begin{document}

\maketitle

\begin{abstract}

We identify and explore differential access to population-level \emph{signaling} (also known as \emph{information design}) as a source of unequal access to opportunity. A population-level signaler has potentially noisy observations of a binary type for each member of a population and, based on this, produces a signal about each member.  A decision-maker infers types from signals and accepts those individuals whose type is high in expectation.  We assume the signaler of the disadvantaged population reveals her observations to the decision-maker, whereas the signaler of the advantaged population forms signals strategically.  We study the expected utility of the populations as measured by the fraction of accepted members, as well as the false positive rates (FPR) and false negative rates (FNR).  

We first show the intuitive results that for a fixed environment, the advantaged population has higher expected utility, higher FPR, and lower FNR, than the disadvantaged one (despite having identical population quality), and that more accurate observations improve the expected utility of the advantaged population while harming that of the disadvantaged one.  We next explore the introduction of a publicly-observable signal, such as a test score, as a potential intervention.  Our main finding is that this natural intervention, intended to \emph{reduce} the inequality between the populations' utilities, may actually \emph{exacerbate} it in settings where observations and test scores are noisy.
\end{abstract}

\section{Introduction}

Settings where personal data drive consequential decisions, at large scale, abound---financial data determine loan decisions, personal history affects bail and sentencing, academic records feed into admissions and hiring. Data-driven decision-making is not reserved for major life events, of course; on a minute-by-minute basis, our digital trails are used to determine the news we see, the job ads we are shown, and the behaviors we are nudged towards.

There has been an explosion of interest recently in the ways in which such data-driven decision-making can reinforce and amplify injustices. One goal of the literature has been to identify the points in the decision-making pipeline that can contribute to unfairness. For example, are \emph{data more noisy or less plentiful} for a disadvantaged population than for an advantaged one? Are the available \emph{data less relevant} to the decision-making task with respect to the disadvantaged population? 
Has the disadvantaged population historically been \emph{prevented or discouraged from acquiring good data profiles} that would lead to favorable decisions? 
Is the decision-maker simply \emph{making worse decisions} about the disadvantaged population, despite access to data that could prevent it?

In this paper, we study \emph{access to population-level signaling} as a source of inequity that, to our knowledge, has not received attention in the literature. We consider settings where the data of individuals in a population passes to a \emph{population-level signaler}, and the signaler determines what function of the data is provided as a signal to a decision-maker. The signaler can serve as an advocate for the population by filtering or noising its individuals' data, but cannot outright lie to the decision-maker; whatever function the signaler chooses to map from individuals' data to signals must be fixed and known to the decision-maker.

Examples of population-level strategic signalers include high schools, who, in order to increase the chances that their students will be admitted to prestigious universities, inflate their grades, refuse to release class rankings~\cite{OS10}, and provide glowing recommendation letters for more than just the best students. Likewise, law firms advocate on behalf of their client populations by selectively revealing information or advocating for trial vs.~plea bargains. Even the choice of advertisements we see online is based on signals about us sold by exchanges, who wish to make their ad-viewing population seem as valuable as possible. 

Our interest in asymmetric information in general and in population-level strategic signaling in particular are inspired by the recent wave of interest in these issues in the economics literature (see Section~\ref{sec:rel_work} for an overview). In particular, the model we adopt to study these issues in the context of inequity parallels the highly influential work on Bayesian persuasion~\cite{KG11} and information design~\cite{infodesign17}.

In order to explore the role that population-level strategic signaling can play in reinforcing inequity, we investigate its impact in a stylized model of university admissions.

We consider a setting in which a high school's information about its students is noisy but unbiased.  Throughout, we call this noisy information \textit{grades}, but emphasize that it may incorporate additional sources of information such as observations of personality and effort, that are also indicative of student quality.  Importantly, all relevant information about student quality is observed directly by the school alone.  
%

The school then aggregates each student's information into a signal about that student that is transmitted to the university.  This aggregation method is called a \textit{signaling scheme}, or informally, a (randomized) mapping from a student's information to a recommendation. 
A school could, for instance, choose to give the same recommendation for all its students, effectively aggregating the information about all students into one statement about average quality.  Or, for example, the school could choose to provide positive recommendations to only those students that it believes, based on its information, to have high ability. 

 The university makes admission decisions based on these recommendations, with the goal of admitting qualified students and rejecting unqualified ones.\footnote{In our simple model, the university does not have a fixed capacity, nor does it consider complementarities between students.} A school might make recommendations designed to maximize the number of their students admitted by the university.  We call such a school {\em strategic}.   Alternatively, a school might simply report the information it has collected on its students to the university directly.  We call such a school {\em revealing}. As is common in economics, we assume that the university knows the signaling scheme chosen by the school (but does not know the realization of any randomness the school uses in its mapping).  One justification typically given for such an assumption is that the university could learn this mapping over time,  as it observes student quality from past years.

As expected, we find that strategic schools with accurate information about their students have a significant advantage over revealing schools, and, in the absence of intervention, strategic schools get more of their students (including unqualified ones) admitted by the university.

A common intervention in this setting is the standardized test.  The university could require students to take a standardized test before being considered for admission, and use test scores in addition to the school's recommendations in an effort to enable more-informed admissions decisions.  Intuitively, the role of the standardized test is that it ``adds information back in'' that was obfuscated by a strategic school in its recommendations, and so one might naturally expect the test to reduce inequity in the admissions process. 
 While such a standardized test does increase the accuracy of admissions decisions, we show that when the test is a noisy estimate of student quality, it may in fact exacerbate the impact of disparities in signaling between schools. 

\paragraph{Summary of contributions} We highlight access to strategic population-level signaling, as studied in the economics literature, as a potential source of inequity. We derive the optimal signaling scheme for a school in Section~\ref{sec:opt_scheme} and compute the resulting school utility and false positive and negative rates in Section~\ref{sec:utilfprfnr}. We then show in Section~\ref{sec:consequences} that disparities in abilities to signal strategically can constitute a non-negligible source of inequity.
In Section \ref{sec: standardized_exam}, we study the effect of a standardized test that students must take before applying to the university, and highlight its limitations in addressing signaling-based inequity. 

\section{Related work}~\label{sec:rel_work}

There is a large literature on individual-level signaling in economics, following on the Nobel-prize-winning work of Spence~\cite{Spence}. The general model there is quite different from our population-level signaling model; in the Spence model, \emph{individuals} (not populations) invest in \emph{costly} (in terms of money or effort) signals whose costs correlate with the individual's type. In that model, equilibria can emerge where high-type individuals are more likely to invest in the signal than low-types, which can result in the signal being useful for admissions or hiring.

Closer to our setting, Ostrovsky and Schwarz~\cite{OS10} study a model in which schools provide noisy information about their students to potential employers. Their focus is on understanding properties of the equilibria of the system; they do not fully characterize the equilibria, they do not consider the role of signaling in compounding inequity, and they do not investigate the impact of interventions like our standardized test. Unlike us, they do not consider the case where the schools have imperfect observations of the students' types. Such work falls into a broader literature on optimal information structures (e.g.,~\cite{RS10}).

The impact of information asymmetries is a common theme in economics today, with key early work including Brocas and Carrillo~\cite{BC08}. Our model of signaling is inspired by the influential work on Bayesian Persuasion~\cite{KG11}, where a persuader (played, in our model, by the school) commits to revealing some fixed function of the types of the population it serves; this revelation is used as the basis of a decision that impacts the welfare of both the decider and the persuader (and the persuader's constituents).  The Bayesian Persuasion model has been applied to a variety of domains, e.g.~\cite{BBM15, RJJX15, XRDT15, KMP14,BDP07, EFGPT12, JM06, AR06}, and generalizations and alternatives to this model have been studied in~\cite{RS10, AB16, GK17, AC16, GK14, KMZL17}. Recent work~\cite{D14, D17, DH16, DH17, EFGPT12, GD13, DIR14} has explored algorithmic aspects of persuasion settings. To our knowledge, ours is the first work to consider access to population-level signaling, Bayesian Persuasion, or information design as a source of inequity.


Recent work on fairness has highlighted a number of objectives that one might wish to enforce when allocating resources to or making decisions about large numbers of individuals. At a high level, these objectives tend to focus either on ensuring group-level fairness~\cite{FFMS15,Kamiran2012, HDF13, HPS16, FSV16, Choul17, KMR17, ZVGG17, KNRW18} or individual-level fairness~\cite{DHPRZ12,JKMR16,KKMPRVW17}. The metrics we study---expected utility, false positive rates and false negative rates---are generally considered to be metrics of group fairness, but they also (coarsely) compare the extent to which similar individuals are being treated similarly.

One very interesting recent paper on fairness~\cite{HC17} does incorporate Spence-style individual-level signaling; in their model, a worker can choose whether and how much to invest in human capital, and this acts as an imperfect signal on whether the worker is qualified. Although their model and its implications are very different from ours, they similarly investigate the impact of upstream interventions on downstream group-level unfairness. Similar notions of individual-level signaling can also be found in~\cite{FV92,CL93}. 


\section{Model}\label{sec:model}
We consider a setting with high schools (henceforth, ``schools''), and a single university. A school has a population of students. Each student $i$ has a binary type $t_i \in \{0,1\}$ 
that represents the quality of the student. The students' types are drawn i.i.d.\ from a Bernoulli distribution with mean $p$; that is, a student has type $1$ w.p.~$p$ and $0$ w.p.~$1-p$. A student's type is private, that is, known to the student but unknown to both the school and the university. The prior $p$ is public and common knowledge to all agents.

A school observes noisy information about the types of each of its students. To formally model this, we assume student $i$ has a grade $g_i\in\{0,1\}$, which is observed by the school but is unknown to the university.

The grade $g_i$ for student $i$ is drawn as follows: 
 $\Pr[g_i = 0 | t_i = 0] = \Pr[g_i= 1 | t_i = 1 ] = q$, for $q \in [1/2,1]$.\footnote{The assumption that $q \geq 1/2$ is without loss of generality; when $q < 1/2$, one can set $q = 1-q, g_i = 1-g_i$ and all results carry through by symmetry.} That is, the student's type is flipped with some probability $1-q$.  As $q$ increases, the grade $g_i$ becomes a more accurate estimate of the student's type $t_i$. The grade $g_i$ is known to the school but {\em not} the university.  The distribution $q$ of the grade, however, is public, i.e., common knowledge to all parties. 

A school has access to a (possibly trivial or uncountably infinite) set of signals $\Sigma$, and commits to a signaling scheme 
mapping grades $g$ to probability distributions over signals in $\Sigma$.
For each student $i$, the university makes an accept/reject decision based on the distribution of the types $p$, the distribution of the grades $q$, and the realization of the signal chosen by the school. The goal of the university is to maximize the quality of the students it accepts.\footnote{There is no notion here of students ``applying'' to the university or not; the university considers \emph{all} students for admission.} In particular, we model the university as having additive utility over the set of students it accepts, with utility $1$ for accepting a student of high type ($t_i=1$), and utility $-1$ for a student with low type ($t_i=0$). We assume that the university has unlimited capacity; therefore, the university accepts exactly those students who induce non-negative expected utility given the common priors and the signal.\footnote{When indifferent, the university accepts the student.}
We measure a school's utility by the expected fraction of its students who are admitted to the university. We note that this choice of utility measures the {\em access to opportunity} (defined as admittance to university) of the school's students.  We  refer to a school as \emph{revealing} if it simply transmits the grade to the university as the signal.
We  refer to a school as \emph{strategic}
 if it employs the optimal strategic signaling scheme, as examined in Section~\ref{sec:opt_scheme}. A strategic school thus maximizes its expected utility.

In several places, we will discuss the distribution of students accepted by the university. To do so, it is useful to introduce the notions of \textit{false positive} and \textit{false negative} rates. The \textit{false positive rate} of a school is the (expected) probability that a student with type $0$ is accepted by the university. The \textit{false negative rate} of a school is the (expected) probability that a student with type $1$ is rejected by the university. 

We introduce several assumptions that restrict our attention to settings of interest.  First, we assume the expected quality of a student is negative, such that the university would reject students without any signal from the school.

\begin{assumption}\label{as: utility_noinfo}
The university's expected utility for accepting any given student, absent any auxiliary information, is negative, i.e., $p - (1-p) <0$, and therefore $p < 1/2.$
\end{assumption}

Next we assume the university's expected utility of accepting a student with a high (resp.\ low) grade is positive (resp.\ negative).

\begin{assumption}\label{as: grades}
The university has non-negative expected utility for accepting a student with a high grade, and negative expected utility for accepting a student with a low grade:
\begin{align*}
&\Pr \left[t = 1 | g = 1 \right] - \Pr \left[t = 0 | g = 1 \right] \geq 0;
\\&\Pr \left[t = 1 | g = 0 \right] - \Pr \left[t = 0 | g = 0 \right] < 0.  
\end{align*}
These can be rewritten as: 
\begin{align*}
& p q - (1-p)(1-q) \geq 0; 
\\& p (1-q) - (1-p) q < 0. 
\end{align*}
\end{assumption}
We note that if the expected utility of accepting a student with a high grade were negative, then none of the school's students would be admitted by the university under any signaling scheme. On the other hand, if the expected utility of accepting a student with a low grade were positive, then the university would always accept every student.\footnote{In fact, this condition is already ruled out by Assumption~\ref{as: utility_noinfo}.}  Thus, this assumption restricts our analysis to the regime in which the utilities of revealing and strategic schools may differ.

The following easy consequence of these assumptions will be useful in our analysis.

\begin{observation}\label{clm: min_q}
Under Assumption \ref{as: utility_noinfo}, Assumption \ref{as: grades} implies $q \geq 1-p.$
\end{observation}


We conclude with the following well-known result (see, e.g., Kamenica and Gentzkow~\cite{KG11}) that an optimal signaling scheme contains, without loss of generality, at most as many signals as there are actions available to the decision-maker.  In our setting, this corresponds to restricting $|\Sigma|=2$ as the university makes an accept/reject decision for each student.

The result, reproduced below for our setting, follows from a revelation-principle type argument.  The idea is to replicate the utilities of a signaling scheme with many signals by first producing a signal 
according to the original scheme and then simply reporting to the university, as a signal in the simplified scheme, 
the action $\sigma^+=$ {\em accept} or $\sigma^-=$ {\em reject} that it would choose to take as a result of seeing the original signal.

\begin{theorem}[Kamenica and Getzkow~\cite{KG11}]\label{clm: rev_principle}
Suppose $\Sigma$ is a measurable (but  potentially uncountable) set with at least two elements. Let $\Sigma'$ be such that $|\Sigma'| = 2$. Given any original signaling scheme mapping to $\Delta\left(\Sigma\right)$,
there exists a new signaling scheme 
mapping to $\Delta\left(\Sigma'\right)$ that induces the same utilities for the school and the university as those induced by the original scheme.
Further, one can write $\Sigma' = \{\sigma^-,\sigma^+\}$ such that a student with signal $\sigma^+$ is accepted by the university with probability $1$, and a student with signal $\sigma^-$ is rejected with probability $1$.
\end{theorem}

When $|\Sigma|=1$, signals carry no information, making mute the question of  access to signaling schemes. Therefore, throughout the paper, we make the assumption that $\vert \Sigma \vert = 2$ and denote its elements by $\Sigma=\{\sigma^+,\sigma^-\}$. This is without loss of generality, by the argument above. 

\section{The impacts of signaling schemes}
The goal of this paper is to highlight the role of access to strategic signaling in creating unequal access to opportunity and explore the intervention of a standardized test as a way to combat this inequity. In order to do so, we first formulate optimal signaling schemes, and then we study their impact on students and their relationship to noisy grades.

\subsection{Optimal signaling scheme}\label{sec:opt_scheme}
We first derive the optimal signaling scheme.  The idea is to pack low-quality students together with high quality students by giving both the {\em accept} signal $\sigma^+$.  A school is limited in the extent to which it can do so, as it must ensure the university obtains non-negative expected utility by accepting all the students who have signal $\sigma^+$.  The following theorem provides the right balance.

\begin{theorem}\label{thm:opt_signal}
The optimal signaling scheme for a school 
is
\begin{align*}
\Pr \left[\sigma^+ |~g = 0 \right] 
&= \frac{p + q - 1}{q - p}\\
\Pr \left[\sigma^+ |~g = 1 \right] &= 1.
\end{align*}
\end{theorem}

\begin{proof}
As per the revelation principle in Theorem~\ref{clm: rev_principle}, we can let $\sigma^+$ be a signal such that all students with that signal are accepted by the university, and $\sigma^-$ a signal such that all students with that signal are rejected. Conditional on $\sigma^+$, we can write the probabilities that a student is of each type as 
\begin{align*}
\Pr &[ t = 1 | \sigma^+] \\
 =& \frac{\Pr \left[ t = 1, \sigma^+ \right]}{\Pr \left[\sigma^+ \right]}\\
=& \Pr[t = 1]\cdot\frac{\Pr[\sigma^+ | t = 1]}{\Pr[\sigma^+]}\\
=& \Pr[t = 1]\cdot\frac{\Pr[\sigma^+ | g = 1]\Pr[g = 1 | t = 1] + \Pr[\sigma^+ | g = 0] \Pr[g = 0 | t = 1]}{\Pr[\sigma^+]}\\
=& p \cdot \frac{q \Pr[\sigma^+ | g = 1] + (1-q) 
\Pr[\sigma^+ | g = 0]}{\Pr \left[\sigma^+\right]}
\end{align*}
and, similarly,
\begin{align*}
\Pr [ t = 0 | \sigma^+] 
&= (1-p) \cdot \frac{(1-q) \Pr[\sigma^+ | g = 1] + q \Pr[\sigma^+ | g = 0]}{\Pr \left[\sigma^+\right]}.
\end{align*}

The university's expected utility when accepting all those students with signal $\sigma^+$ is non-negative if and only if such a student is at least as likely to be of type $1$ as of type $0$, that is, $\Pr[t = 0 | \sigma^+] \leq \Pr[t = 1 | \sigma^+ ]$.
Plugging in and rearranging, this gives the constraint
\begin{align*}
 \Pr[\sigma^+ | g = 0] &\cdot \left( q (1-p) - p (1-q)  \right)\\
 &\leq \Pr[\sigma^+ | g = 1] \cdot \left( p q - (1-q) (1-p)  \right).
\end{align*}
Recall that 
$q ( 1 - p)  - p (1-q) > 0$ by Assumption~\ref{as: grades}, and thus the constraint can be rewritten as 
\begin{align*}
\Pr[\sigma^+ | g = 0]  &\leq \frac{ p q -  (1-q) (1-p)}{ q ( 1 - p) - p (1-q) } \cdot \Pr[\sigma^+ | g = 1] \\
&= \frac{p + q - 1}{q - p}\cdot \Pr[\sigma^+ | g = 1].
\end{align*}
The school's expected utility is
\[\Pr[\sigma^+] = \Pr[\sigma^+ | g  = 0]\Pr[g = 0] + \Pr[\sigma^+ | g = 1] \Pr[g = 1].\]
Since $\Pr[\sigma^+ | g = 1]$ is unconstrained, the school's utility is maximized by setting it to $1$. The school's utility is, similarly, maximized by maximizing the value of $\Pr[\sigma^+ | g  = 0]$, which, given the constraint, occurs by setting
\begin{align*}
\Pr[\sigma^+ | g = 0]  
&= \frac{p + q - 1}{q - p}\cdot \Pr[\sigma^+ | g = 1]\\
& = \frac{p + q - 1}{q - p}.\qedhere
\end{align*}
\end{proof}

\subsection{School's utility, false positive and false negative rates}
\label{sec:utilfprfnr}

In this section, we calculate the expected utility, false positive, and false negative rate achieved by a school, depending on the accuracy of its grades and whether it uses the optimal strategic signaling scheme when transmitting information about its students to the university. These lemmas will form the basis of our evaluation of the impacts of strategic signaling, in Section~\ref{sec:consequences}. Recall that we refer to a school that does not strategically signal and instead transmits its raw grades to the university as \emph{revealing}.

The proofs of the following Lemmas follow by direct calculations.  We provide an exposition of the more involved calculations of Lemmas~\ref{lem:signal-util} and~\ref{lem:signal-fairness} in Appendix~\ref{app: proofs_noexam}.

\begin{lemma}[Revealing school's utility]\label{lem:reveal-util}
The expected utility $U_r(p,q)$ of a revealing school is 
\[
U_r(p,q) = p q + (1-p)(1-q).
\]
For the special case of a revealing school with accurate grades (when $q = 1$), we have
\[
U_r(p,1) = p.
\]
\end{lemma}

A revealing school gets exactly the students with high grades accepted, as per Assumption~\ref{as: grades}; in particular, a $q$ fraction of high-type students will have a high grade and be accepted, while a $(1-q)$ fraction of the low-type students will be accepted.

\begin{lemma}[Strategic school's utility]~\label{lem:signal-util}
A school's expected utility $U_s(p,q)$ when it signals strategically is given by
\begin{align*}
U_s(p,q) = 1 + (p + q - 2 p q) \cdot \frac{2 p - 1}{q - p}.
\end{align*}
For the special case of a strategic school with accurate grades (when $q = 1$), we have
\[
U_s(p,1) = 2p.
\]
\end{lemma}

 A school that signals strategically gets exactly those students with a signal of $\sigma^+$ accepted, as per the revelation principle argument of Theorem~\ref{clm: rev_principle}; a student with a high grade will be accepted with probability $\Pr \left[\sigma^+ |~g = 1 \right]$ and a student with a low grade with probability $\Pr \left[\sigma^+ |~g = 0 \right]$, with the probabilities chosen according to Theorem~\ref{thm:opt_signal}.

\begin{lemma}[Revealing school's FPR/FNR]\label{lem:reveal-fairness}
When a school is revealing, the false positive rate is given by 
\begin{align*}
FPR_r(p,q) = 1 - q
\end{align*}
and the false negative rate by
\begin{align*}
FNR_r(p,q) = 1 - q.
\end{align*}
For the special case of a revealing school with accurate grades (when $q = 1$), we have $FPR_r(p, 1) = FNR_r(p, 1) = 0$.
\end{lemma}
In the case of a revealing school, a low-type (resp.\ high-type) student obtains a low (resp.\ high) grade and gets
rejected (resp.\ accepted) with probability $1 - q$, i.e., if the grade does not match the type.

\begin{lemma}[Strategic school's FPR/FNR]~\label{lem:signal-fairness}
When a school signals strategically, the false positive rate is given by 
\begin{align*}
FPR_s(p,q)= 1-q + q \cdot \frac{p + q - 1}{q - p } 
\end{align*}
and the false negative rate by
\begin{align*}
FNR_s(p,q) =(1-q) \frac{ 1- 2p}{ q -p  }.
\end{align*}
For the special case of a strategic school with accurate grades (when $q = 1$), we have $FPR_s(p,1) = \frac{p}{1-p}$
and 
$FNR_s(p,1) = 0.$
\end{lemma}

In the case of a school that signals strategically according to Theorem~\ref{thm:opt_signal}, a low-type student gets accepted with probability $\Pr \left[\sigma^+ |~g = 1 \right] = 1$ if his grade is $1$ (which occurs with probability $1 - q$), and probability $\Pr \left[\sigma^+ |~g = 0 \right]$ if his grade is $g = 0$ (which occurs with probability $q$).  On the other hand, a high-type student gets rejected when his signal is $\sigma^-$; because $\Pr \left[\sigma^+ |~g = 1 \right] = 1$, this happens only when $g = 0$ and the signal is $\sigma^-$, i.e. with probability $\Pr \left[\sigma^- |~g = 0 \right] \Pr[g=0 | t=1]$. 

\paragraph{Remark}
While we chose to focus on average population (i.e., school) utility in this paper, because of space constraints, one can use these derivations of FRP and FNP to calculate the welfare of subpopulations, such as low-type students at a revealing school, which then implies population-level utility comparisons as well.  One interesting observation is that, using the above Lemmas and Assumptions~\ref{as: utility_noinfo} and~\ref{as: grades}, one can see that the FPR of a strategic school is {\em larger} and the FNR {\em smaller} than that of a revealing school.  Thus, while it is intuitively obvious that low-type students prefer a strategic school, these calculations show that high-type students also prefer a strategic school (and the preference is strict unless the assumptions hold with equality).

\subsection{Consequences of strategic signaling for access to opportunity}~\label{sec:consequences}

In this section, we quantify the impact of access to strategic signaling and its interaction with accuracy of the information (grades) on which the signals are based. We study both the resulting expected utility of a school as well as the resulting acceptance rates of both types of students.  We find that the ability to strategically signal always has a positive (although bounded) impact, increasing students' acceptance rates and the school's expected utility.  The benefit of strategic signaling for both students and the school improves (boundedly so) with the accuracy of the grades, whereas a revealing school and its students receive (potentially dramatically) higher expected utility from noisy grades. The following theorem is a direct consequence of Lemmas~\ref{lem: monotonicity_results},~\ref{lem: effect_grades_fix_revealing},~\ref{lem: effect_grades_fix_strat},~\ref{lem: effect_strategic_fix_acc_g}, and \ref{lem: effect_strategic_fix_noisy_g} in the Appendix.

\begin{theorem}\label{thm: utility_comparison}
For all $p < 1/2$ and $q>q'\geq 1-p$, the following hold: 
\begin{itemize}
\item accuracy in grades benefits strategic schools,
$$\frac{1}{1 - p} U_s(p,q')\geq U_s(p,q) \geq U_s(p,q');$$ 
\item strategic schools have higher expected utility than revealing schools,
$$2 U_r(p,q) \geq U_s(p,q) \geq U_r(p,q);$$ 
\item and accuracy in grades harms revealing schools,
$$2(1-p) U_r(p,q)\geq U_r(p,q') \geq U_r(p,q).$$ 
\end{itemize}
Further, all above bounds are tight for some $q, q'$. 
\end{theorem}

We see that, perhaps counter-intuitively, adding noise to the grades can help a revealing school get more students admitted, up to a point\footnote{A similar observation in a somewhat different setting was made in work of Ostrovsky and Schwarz~\cite{OS10}.}. This follows from the fact that adding noise to the grade increases the number of students with a high grade overall, by Assumption~\ref{as: utility_noinfo}, as there are more low-type students (whose representation increases as grade accuracy decreases) than high-type students (whose representation decreases as grade accuracy decreases). Adding noise to grades is, however, a blunt instrument, in that it drives up both false negatives and false positives (see Lemma~\ref{lem:reveal-fairness}), which limits its utility benefits. The ability to signal strategically is more subtle, driving up false positives (and expected utility), at no cost of false negatives. The power of strategic signaling is maximized when schools have access to highly accurate grades. Accurate information, the ability to control the noise level of that information, and, most notably, the ability to strategically signal about that information, therefore constitute powerful drivers of unequal access to opportunity in settings where key information is transmitted to a decision-maker on behalf of a population.

We can derive comparisons resulting in similar insights for the false positive and false negative rates of revealing and strategic schools (see Appendix~\ref{app: false_rates_noexam}).

\section{Intervention: Standardized Test}\label{sec: standardized_exam}

The prior sections show that unequal access to strategic signaling can result in unequal access to opportunity.  This is driven by high error rates for students accepted from schools with signaling technologies and/or noisy grades.  The university has a vested interested in decreasing this error rate as it harms the university's utility.  In addition, an outside body or the university itself might be concerned about the resulting unequal access to opportunity. In this section, we explore the impact of a common intervention: the standardized test.  While availability of a test score certainly can only improve the expected utility of the university,\footnote{This is because the expected utility of the university from strategic schools without test scores is zero, and so can only increase. For revealing schools, the university gets strictly more information with test scores and hence more utility.} we find that it has an ambiguous effect on the inequity.  In particular, for a large range of parameter settings, the introduction of a test can {\em increase} the inequality in access to opportunity.

\subsection{Augmented model}
Throughout this section, we augment the model of Section~\ref{sec:model} to add the requirement that each student must take a test, and the results of that test are visible both to the student's school and to the university. (The school may then incorporate the test results into its subsequent strategic behavior.)

We model the test score $s_i\in \{0,1 \}$ of student $i$ as a noisy estimate of $t_i$, conditionally independent from the grade $g_i$, obtained as follows: 
 $\Pr[s_i = 0 | t_i = 0] = \Pr[s_i= 1 | t_i = 1 ] = \delta$, for $\delta \in [1/2,1]$.\footnote{The assumption that $\delta \geq 1/2$ is, as with our analogous assumption about the grades, without loss of generality.} The score $s_i$ is public, i.e., the school and the university both observe it. 

A school has access to a set of signals $\Sigma$ as before, but now can design a signaling scheme $\sigma:\{0,1\}\times\{0,1\}\rightarrow\Delta(\Sigma)$ that is a function of both the student's grade and his test score; i.e., the school designs $\Pr \left[\sigma |~g_i,s_i \right]$ for $\sigma \in \Sigma$. The university again makes accept/reject decisions that maximize its expected utility, but now the university has access to the test score $s_i$ and its distribution $\delta$ as well as the signal and the distributions $p$ and $q$.  As before, a {\em strategic} school chooses a signaling scheme that maximizes the fraction of students accepted whereas a {\em revealing} school simply transmits the grade to the university as the signal.

As in Section~\ref{sec:model}, we introduce an assumption controlling the noise $\delta$ of the test.

\begin{assumption}\label{as: util_SAT}
The university has non-negative expected utility for accepting a student with a high test score, and negative expected utility for accepting a student with a low test score:
\begin{align*}
 0 &\leq p \delta -  (1-p) (1-\delta)  
\\  0 &> p (1-\delta) -  (1-p) \delta .
\end{align*}
\end{assumption}

We note that if the expected utility of accepting a student with a high test score were negative, or the expected utility of accepting a student with a low test score were positive, then in the absence of signals, the university would always accept either none or all of the students. Note that regimes when the standardized test is uninformative on its own but becomes informative when coupled with grades may still be interesting. However, even under  Assumption~\ref{as: util_SAT}, which excludes certain parameter ranges from consideration, we have a rich enough model to illustrate our main findings. In Appendix~\ref{app: relaxed_assumption} we show how to relax this assumption, and how doing so affects the optimal signaling scheme.  

The following consequence will be useful in our analysis.

\begin{observation}
Under Assumption~\ref{as: utility_noinfo}, Assumption~\ref{as: util_SAT} implies
\begin{align*}
\delta \geq 1-p.
\end{align*}
\end{observation}

Fixing $p$, we denote by $u_{q,\delta}(g,s)$ the expected utility the university derives from admitting a student with score $s$ and grade $g$:
\begin{align*}
u_{q,\delta}(g,s):=\Pr[t_i=1|~g,s]-\Pr[t_i=0|~g,s].
\end{align*}
When $\delta=q=1$, $u_{q,\delta}(s,g)$ is not defined for $s\not=g$ as in this case $s$ and $g$ are perfectly correlated.
For notational convenience, we define $u_{q,\delta}(s,g)=-1$ in these cases.

\begin{lemma}\label{lem: SAT+gradesconditions}  
Assumptions~\ref{as: grades} and~\ref{as: util_SAT} together imply that 
the university receives non-negative expected utility from accepting a student with both a high grade and a high score, and negative expected utility from a student with both a low grade and a low score:
\[u_{q,\delta}(1,1) \geq 0 > u_{q,\delta}(0,0).\]
This can be rewritten as 
\begin{align*}
& p q \delta - (1-p)(1-q) (1-\delta) \geq 0;
\\& p (1-q) (1-\delta) - (1-p) q \delta < 0.
\end{align*}
\end{lemma}

Theorem~\ref{clm: rev_principle} (the revelation principle) also holds in this setting, and so we assume for the remainder of this section that $\Sigma = \{\sigma^-, \sigma^+\}$, without loss of generality.

\subsection{Optimal signaling}

We first derive the optimal strategic signaling scheme.  Again, a school would like to  pack low-quality students together with high quality students, but is now limited in its ability to do so by their test scores.  If the expected utility the university receives from a student with a high grade but low test score is negative ($u_{q,\delta}(1,0)<0$), then this student (and in fact any student with a low test score) will be rejected regardless of the signal from the school.  Otherwise ($u_{q,\delta}(1,0)\geq0$), the school can signal to the university to accept such a student, and can additionally pack in some low-grade-low-score students, subject to maintaining non-negative expected utility for the university.

\begin{theorem}~\label{thm:opt-signaling-exam}
The optimal signaling scheme for a school with access to grades and a test score, under Assumption~\ref{as: util_SAT}, is
\begin{align*}
\Pr \left[\sigma^+ |~g = 1,s = 1 \right] &= 1\\
\Pr \left[\sigma^+ |~g = 0,s = 1 \right] &= 1\\
\Pr \left[\sigma^+ |~g = 1,s = 0 \right] &=
\begin{cases}
1, & \text{if}\ u_{q,\delta}(1,0) \geq 0 \\ 
0, & \text{if}\ u_{q,\delta}(1,0) < 0
\end{cases}\\
\Pr \left[\sigma^+ |~g = 0,s = 0 \right] &= 
\begin{cases}
\frac{p q (1-\delta) - (1-p) (1-q) \delta }{ (1-p) q \delta - p (1-q) (1-\delta)  }, & \text{if}\ u_{q,\delta}(1,0) \geq 0\\ 
0, & \text{if}\ u_{q,\delta}(1,0) < 0
\end{cases}
\end{align*}
\end{theorem}
We defer the proof to Appendix~\ref{app: proofs_withexam}.

\subsection{School's utility, false positive and false negative rates}

In this section, we calculate the expected utility achieved by both a strategic school and a revealing school as a function of the type distribution, the accuracy of its grades, and the accuracy of the  standardized test score.  We defer all proofs to Appendix~\ref{app: proofs_withexam}.

For a revealing school, the university always accepts high-grade high-score students.  If high grades are more informative than low test scores (that is, if $u_{q,\delta}(1,0)\geq0$, which depends on $p$ as well as $q$ and $\delta$ and happens, for instance, if $p=1/4$, $q=9/10$, and $\delta=7/10$), then the university also accepts students with low test scores, benefiting the school.  Alternatively, if high test scores are more informative than low grades (i.e., $u_{q,\delta}(0,1)\geq 0$), then the university also accepts students with low grades.  These conditions provide additional boosts to the utility of a revealing school. 

\begin{lemma}[Revealing school's utility]~\label{lem:revealing-util-exam}
The expected utility $U_r(p,q,\delta)$ of a revealing school with access to grades and a test score is 
\begin{align*}
U_r(p,q,\delta) 
&= p q \delta + (1-p)(1-q)(1-\delta)
\\&+ \mathbbm{1} \left[ u_{q,\delta}(1,0) \geq 0 \right] \left(p q (1-\delta) + (1-p)(1-q) \delta  \right)
\\&+ \mathbbm{1} \left[ u_{q,\delta}(0,1) \geq 0 \right] \left(p (1-q) \delta + (1-p) q (1-\delta) \right).
\end{align*}
For the special case of a revealing school with accurate grades (when $q = 1$), we have
\[
U_r(p,1,\delta) = p.
\]
\end{lemma}

As illustrated in Figure~\ref{fig:utilities}, for fixed $p$ and $\delta$, $U_r(p,q,\delta)$ may not be a decreasing function of $q$. In fact, when $q$ is small enough, the grades are completely uninformative and the university only admits students with a test score of $1$. In that regime, the expected utility for a revealing school is therefore constant in $q$. For intermediate values of $q$, the grades are still uninformative on their own but are informative coupled with a high standardized test score; at this point, only students with both a score and a grade of $1$ get admitted by the university, and the school's expected utility suddenly drops when compared to smaller $q$. The school's expected utility in that regime is increasing in $q$ as, under Assumption~\ref{as: util_SAT}, increasing the value of $q$ increases the fraction of students with both high scores and high grades.
 Finally, when $q$ is large enough, the grades are significant enough on their own that only students with high grades are admitted; this leads to a jump in expected utility compared to the intermediate regime. In this regime for high values of $q$, the school's expected utility is decreasing as a result of the fact that increasing the value of $q$ now decreases the number of students with a high grade by Assumption~\ref{as: utility_noinfo}, as seen in Section~\ref{sec:consequences}. 

\begin{lemma}[Strategic school's utility]~\label{lem:signaling-util-exam}
The expected utility $U_s(p,q)$ when a school signals strategically and $u_{q,\delta}(1,0) <0$  is
\[
U_s(p,q,\delta) = p \delta + (1-p) (1-\delta); 
\]
when $u_{q,\delta}(1,0) \geq 0$, the expected utility is
\begin{align*}
U_s(p,q,\delta)
&= \left(1 - p (1-q) (1-\delta) - (1-p)q\delta \right)
\\&+ \left( p (1-q) (1-\delta) + (1-p) q \delta \right) \frac{p q (1-\delta) - (1-p) (1-q) \delta }{ (1-p) q \delta - p  (1-q) (1-\delta)}. 
\end{align*}
For the special case of a strategic school with accurate grades (when $q = 1$), we have
\[
U_s(p,1,\delta) = 1 - \delta + p.
\]
\end{lemma}

The expected utility of a strategic school is, unsurprisingly, monotone in $q$ (as illustrated in Figure \ref{fig:utilities}), as higher-quality information about its students' types allows the school to signal more effectively. For small and intermediate values of $q$ (i.e., insignificant grades), the university bases admission decisions solely on the standardized test score and only admits students with a score of $1$ (it has positive expected utility from doing so, by Assumption \ref{as: util_SAT}); in this regime, a strategic school's expected utility is hence constant. When $q$ becomes large enough, i.e., when the grades are significant enough, the university starts having positive expected utility from admitting students with a high grade even if they have a low score, and the school can start bundling these students together with the high score students, leading to a jump in its expected utility.
The plotted parameters for the figures are chosen to satisfy Assumptions~\ref{as: utility_noinfo}, \ref{as: grades} and \ref{as: util_SAT}; the discontinuities occur at $q$ such that $u_{q,\delta}(0,1) = 0$ and $u_{q,\delta}(1,0) = 0$.



We also calculate the false positive and false negative rates of strategic and revealing schools; we defer this derivation to Appendix~\ref{app: false_rates_withexam}.

\subsection{Impact of Standardized Test}



With a perfect standardized test or, in fact, a sufficiently good one, (i.e., high enough $\delta$), it is not hard to see that the university accepts exactly those students with a high test score from strategic as well as revealing schools. Thus, no matter the accuracy of the grades or distribution of types, the standardized test results in equal expected utility, and hence equal access to opportunity, for revealing and strategic schools (see Appendix~\ref{app: ratio_utilities_with_exam} for details).  Similarly, if grades are accurate (i.e., $q=1$), then a revealing school's expected utility is fixed at $p$ whereas a strategic school's expected utility is only diminished (from $2p$ without the test) by the extra constraints introduced by a standardized test.  Thus, in this case as well, a standardized test decreases the inequality between the utilities of a strategic and a revealing school, making the ratio of utilities less than $2$ (see Appendix~\ref{app: ratio_utilities_with_exam} for details). 


Figure~\ref{fig:truthful_vs_strategic} plots $U_s(p,q)/U_r(p,q)$, with and without test scores, as a function of $q$, for $p = 0.35$ and different values of $\delta$. The form of the utility ratio between a strategic and a revealing school in the absence of a test score follows from the fact that both utilities are continuous, and that the expected utility of a strategic school increases while that of a revealing school decreases in $q$, as we have seen in Section~\ref{sec:consequences}. 
The form in the presence of a test score can be explained as follows. First, when in the regime of small values of $q$, only students with a high standardized test score are admitted by the university, in which case admission decisions do not depend on how the schools act and both the strategic school and the revealing school have the same expected utility, leading to a ratio of $1$. For intermediate values of $q$, we have previously discussed that the utility for a strategic school remains constant (the university still has positive utility for students with a score of $1$ and the strategic school can bundle all such students together, regardless of grade), while the utility for a revealing school suddenly drops (only students with both a high grade and a high score are admitted) and is increasing in $q$, explaining the sudden drop in ratio of utilities at the change of regime, and the decreasing monotonicity of the ratio in $q$ within the intermediate regime. When $q$ becomes large enough, we have seen that both the revealing and the strategic school experience a jump in utilities, which explains the second discontinuity in the ratio of utilities. Because the revealing school has significantly lower utility than the strategic school for intermediate values of $q$, the relative jump in the utility of a revealing school is higher than the relative jump in utility of a strategic school. 

Interestingly, we observe that the introduction of a standardized test does not always decrease inequity. For noisy grades, when the test score is also sufficiently noisy, the test may have the effect of increasing the ratio of utilities between a strategic school and a revealing school.  This is clearly illustrated in Figure~\ref{fig:truthful_vs_strategic}, where the curve with test scores sometimes lies above that without a test. Some intuition for this result is as follows. In the regime for intermediate values of $q$, as $q$ becomes more and more inaccurate, the ratio of utilities in the presence of a standardized test increases and eventually overtakes the ratio in the absence of a standardized test (which decreases to $1$ as the grades become more inaccurate). In the regime for high values of $q$, the university admits students with a high grade only, independently of what their standardized test scores are; therefore, the utility of a revealing school is the same with or without a standardized test. On the other hand, when the standardized test score becomes more inaccurate, the strategic school can take advantage of the noise in said score to bundle in more students than if there was no standardized test: the university loses in utility from accepting unqualified students with high scores, but at the same time gains in utility from accepting qualified students with low scores, allowing a strategic school to bundle more students when compared to the case with no standardized test. As $\delta$ decreases and the standardized test becomes less and less accurate, a strategic school starts losing fewer high-score students  to rejection than it gains in admitted low-score students, and its utility increases.

\section{Further discussion and future directions}
Our paper, in introducing the study of inequity induced by population-level signaling, raises a number of directions for future work. We discuss a few of them here.

First, one might be interested in enriching the model of the standardized test intervention. For example, there could be asymmetries in how students from different schools perform on the standardized test. One might imagine students at an advantaged school might be better prepared for the test (e.g., by investment in expensive test-prep courses), giving them an edge in the form of an increased probability of performing well on the test. Suppose, for example, that high-type students in an advantaged school had a higher probability of passing the test than high-type students at a disadvantaged school. In such a situation, more high-type students from the advantaged school would be admitted by the university, and, as the utility for the university to accept high score students increased, the advantaged school could also bundle a larger number of low-type students with its high-type students. That is, a jump in high-types' exam performance increases students' utilities at that school, even for low types; this effect could further exacerbate disparities between and an advantaged and a disadvantaged school. An interesting question could be to quantify how much  disparities between schools would increase in such a setting. One could also analyze other variants of advantage on the exam, such as an increased probability of passing both for high-types and for low-types. 

One might also imagine that students in an advantaged school might have access to more resources and could take the standardized test several times, while students in a disadvantaged school could only take the test once. When only the highest test score is reported to the university (as is common in practice for university admissions in the United States), it can be seen that this reduces to the situation described above, in which students in each school take the standardized test exactly once, but students in the advantaged school have a higher probability of passing. An extreme case of such a situation would be when the advantaged school's students could take the test enough times that they would pass with a probability approaching $1$; in such a case, a test score from the advantaged school would be meaningless to the university. On the other hand, the test would still be significant for the disadvantaged school, and could have the effect of reducing the number of its students that are accepted, further increasing disparities between schools. A natural question would be to quantify such disparities for intermediate values of the number of times that an advantaged-school student can take the standardized test.

Finally, throughout the paper, we assume that the university has unlimited capacity and is willing to accept every student that provides it with non-negative expected utility. One might ask what would happen if the university had a limited capacity. The university might then rank students as a function of their school of origin, their signal, and their test score (in the presence of a standardized test), and only accepted the highest-ranked students. If an advantaged school had the ability to make its students look better than a disadvantaged school (for example, an advantaged school might have more accurate grades and have a higher ability to strategically signal), then the advantaged school could guarantee that some of its students would get first pick by the university, to the detriment of a disadvantaged school---which would only have access to the (possibly small) remaining capacity. A natural direction would be to understand how much of an effect this limited capacity setting can have on inequity.

\begin{figure}[H]
\centering
    \begin{subfigure}{0.45\textwidth}
    \centering
        \includegraphics[width=1\linewidth]{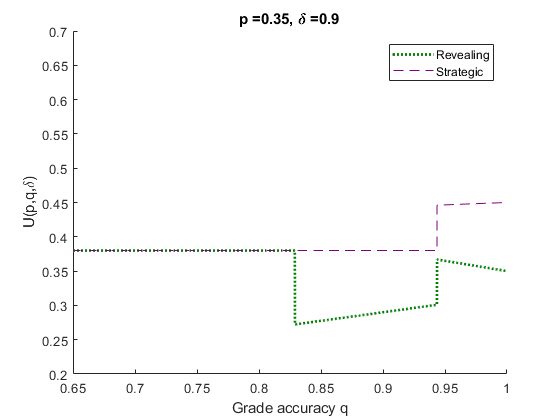}
        \caption{$\delta = 0.90$}
    \end{subfigure}
    \begin{subfigure}{0.45\textwidth}
    \centering
        \includegraphics[width=1\linewidth]{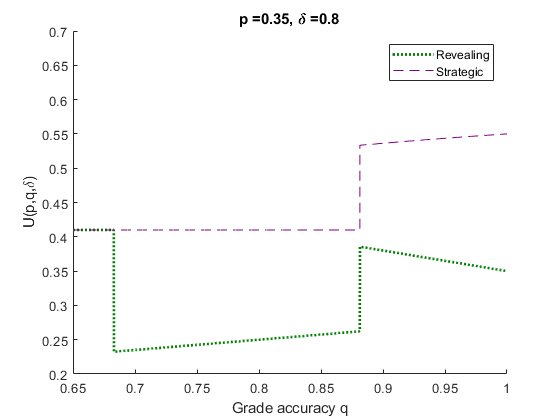}
        \caption{$\delta = 0.80$}
    \end{subfigure}
    \begin{subfigure}{0.45\textwidth}
    \centering
        \includegraphics[width=1\linewidth]{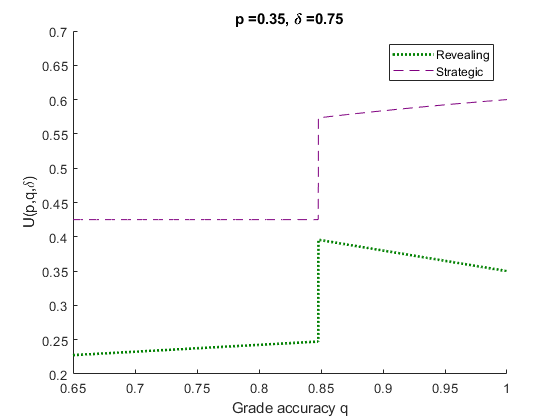}
        \caption{$\delta = 0.75$}
    \end{subfigure}
    \begin{subfigure}{0.45\textwidth}
    \centering
        \includegraphics[width=1\linewidth]{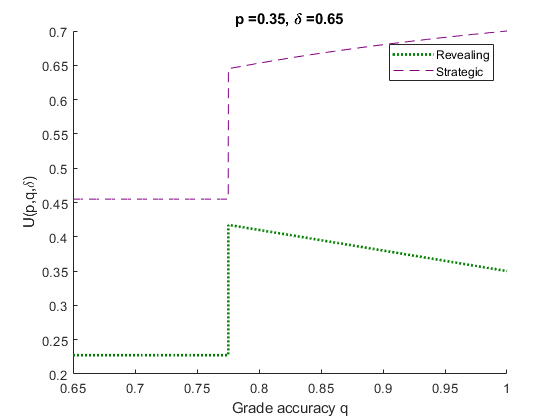}
        \caption{$\delta = 0.65$}
    \end{subfigure}

\caption{Strategic school utility $U_s(p,q,\delta)$  and revealing school utility $U_r(p,q,\delta)$ as a function of the grade accuracy $q$, for average student type $p = 0.35$. We observe that the expected utility may be non-monotone in $q$. 
}
\label{fig:utilities}
\end{figure}

\begin{figure}[H]
\centering
    \begin{subfigure}{0.45\textwidth}
    \centering
        \includegraphics[width=1\linewidth]{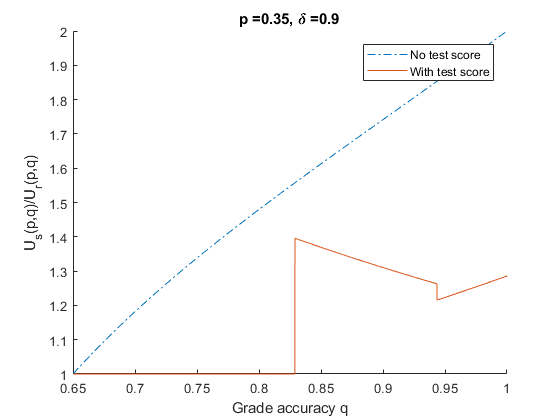}
        \caption{$\delta = 0.90$}
    \end{subfigure}
    \begin{subfigure}{0.45\textwidth}
    \centering
        \includegraphics[width=1\linewidth]{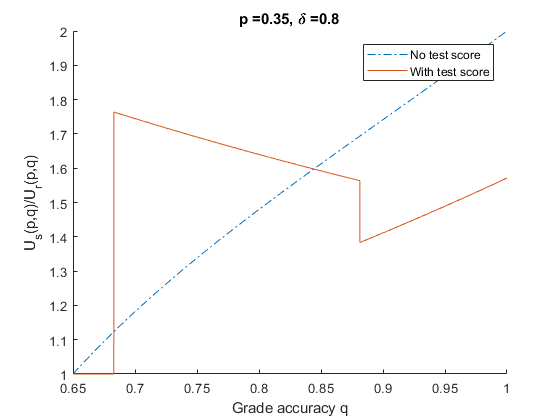}
        \caption{$\delta = 0.80$}
    \end{subfigure}
    \begin{subfigure}{0.45\textwidth}
    \centering
        \includegraphics[width=1\linewidth]{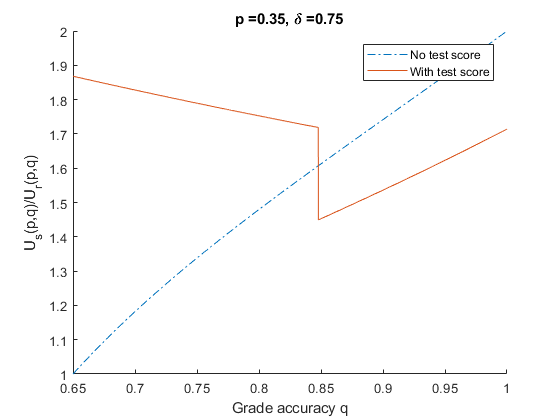}
        \caption{$\delta = 0.75$}
    \end{subfigure}
    \begin{subfigure}{0.45\textwidth}
    \centering
        \includegraphics[width=1\linewidth]{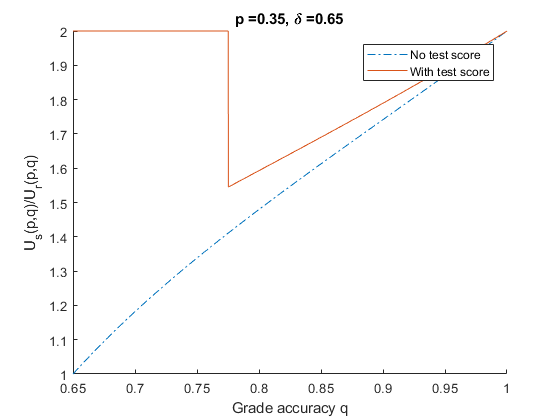}
        \caption{$\delta = 0.65$}
    \end{subfigure}
    
\caption{The ratio $U_s(p,q)/U_r(p,q)$ of utilities of a strategic school vs.~a revealing school, as a function of the grade accuracy $q$, with and without test score. We observe that the test score intervention may increase inequality.}
\label{fig:truthful_vs_strategic}
\end{figure}

\bibliographystyle{abbrv}
\bibliography{references} 

\newpage
\appendix

\section{False positive and negative rates -- Model without standardized test}\label{app: false_rates_noexam}
\begin{theorem}\label{thm: FPR+FNR_comparison}
For all $p < 1/2$ and $q>q'\geq 1-p$, the following hold: 
\begin{itemize}
\item increasing grade accuracy increases the $FPR$ and decreases the $FNR$ of a strategic school,
$$\frac{1}{1 - p} FPR_s(p,q')\geq FPR_s(p,q) \geq FPR_s(p,q'),$$ 
$$FNR_s(p,q') \geq FNR_s(p,q);$$
\item signaling increases the $FPR$ and decreases the $FNR$ as compared with revealing, 
$$FPR_s(p,q) \geq FPR_r(p,q)$$ 
$$\frac{1-2p}{1-p} FNR_r(p,q) \leq FNR_s(p,q) \leq FNR_r(p,q);$$
\item increasing grade accuracy decreases both the $FPR$ and the $FNR$ of a revealing school
$$FPR_r(p,q') \geq FPR_r(p,q),$$
$$FNR_r(p,q') \geq FNR_r(p,q).$$
\end{itemize}
Further, all above bounds are tight for some $q, q'$.
\end{theorem}

\begin{proof}
This is a direct consequence of Lemmas~\ref{lem: monotonicity_results},~\ref{lem: effect_grades_fix_revealing},~\ref{lem: effect_grades_fix_strat},~\ref{lem: effect_strategic_fix_acc_g}, and \ref{lem: effect_strategic_fix_noisy_g}.
\end{proof}

\section{Missing Lemmas and Proofs - Model Without Standardized Test}\label{app: proofs_noexam}


\begin{proof}[Proof of Lemma~\ref{lem:signal-util}]
\begin{align*}
U_s(p,q) 
=& \Pr[\sigma^+ | g = 1] \left(pq + (1-p) (1-q) \right) \\
&+ \Pr[\sigma^+ | g = 0]  \left(p (1-q) + q (1-p) \right)
\\=&  (1 + 2p q - p -q) + \left(p + q - 2 p q \right) \cdot \frac{p + q - 1}{q - p}
\\=& 1 + (p + q - 2 p q) \cdot \frac{2 p - 1}{q - p}. \qedhere
\end{align*}
\end{proof}

\begin{proof}[Proof of Lemma~\ref{lem:signal-fairness}]
\begin{align*}
FPR_s(p,q)
=& \Pr[\sigma^+ | g = 1] \Pr[g=1 | t=0] \\
&+ \Pr[\sigma^+ | g = 0]  \Pr[g=0 | t=0] 
\\=& 1-q + q \cdot \frac{p + q - 1}{q - p } 
\end{align*}
\begin{align*}
FNR_s(p,q)=& \left(1-\Pr[\sigma^+ | g = 1]\right) \Pr[g=1 | t=1] \\
&+ \left(1-\Pr[\sigma^+ | g = 0] \right) \Pr[g=0 | t=1]
\\=& (1-q) \left( 1-\frac{p + q - 1}{q - p }  \right) 
\\=&(1-q) \frac{ 1- 2p}{ q -p  }.\qedhere
\end{align*}
\end{proof}

\begin{lemma}\label{lem: monotonicity_results}
Suppose $p < 1/2$. Then $U_s(p,q)$ and $FPR_s(p,q)$ are increasing functions of $q \in [1-p,1]$. On the other hand, $U_r(p,q)$, $FPR_r(p,q)$, $FNR_r(p,q)$ and $FNR_s(p,q)$ are decreasing functions of $q \in [1-p,1]$.
\end{lemma}

\begin{proof} We first consider the expected utility of a strategic school. 
\begin{align*}
\frac{\partial U_s}{\partial q}\left(p,q\right) 
&= \frac{2p-1}{(q-p)^2} \left((1-2p) (q-p) - (p + q - 2pq) \right) 
\\&= \frac{2p-1}{(q-p)^2} \left(q-p-2pq + 2p^2 -p -q + 2pq\right) 
\\& = 2 \frac{2p-1}{(q-p)^2} p (p-1) 
\\&> 0,
\end{align*}
since $p < 1$ and by Assumption~\ref{as: grades}, $2p - 1 < 0$. Therefore, $U_s(p,q)$ is increasing in $q$.

We next consider the FPR of a strategic school.
\begin{align*}
\frac{\partial FPR_s}{\partial q}\left(p,q\right) = \frac{2p (q-p) - (2pq - p)}{(q-p)^2} = \frac{p - 2p^2}{(q-p)^2} > 0
\end{align*}
as $p < 1/2$ implies $p-2p^2 = p(1-2p) > 0$.

$U_r(p,q) = p q + (1-p) (1-q) = (2p -1) q + 1 - p$ is decreasing in $q$ as $2p - 1 < 0$. $FPR_r(p,q) = FNR_r(p,q) = 1-q$ are immediately decreasing in $q$. $FNR_s(p,q) = (1-q) \frac{1-2p}{q-p}$ is decreasing in $q$ as $\frac{1-q}{q-p}$ is decreasing in $q$ and $1 - 2p > 0$.
\end{proof}

\begin{lemma}\label{lem: effect_grades_fix_revealing}
For a revealing school, the impact on expected utility of moving between noisy and accurate grades is quantified by
\[\frac{1}{2(1 - p)} \leq \frac{U_r(p, 1)}{U_r(p, q)} \leq 1.\]
A revealing school maximizes its expected utility by setting $q = 1 - p$.

The impact on the FPR and the FNR, when $q \neq 1$, is quantified by
\[FPR_r(p,1) = FNR_r(p,1) = 0,~FPR_r(p,q) = FNR_r(p,q) = 1-q.\]

Further, all above bounds are tight for some $q$. 
\end{lemma}

\begin{proof}
For a revealing school, $U_r(p,q) = p q + (1-p) (1-q)$. $U_r(p,q) = p q + (1-p) (1-q) = q (2p - 1) + (1 - p)$ is a decreasing function of $q$, so under Assumption~\ref{as: grades} that $p q \geq (1-p) (1-q)$, a revealing school's expected utility is maximized when $p q = (1-p) (1-q)$, i.e., when $q = 1-p$ and minimized when $q=1$. It is therefore the case that
\[
\frac{U_r(p,1)}{U_r(p,q)} \geq \frac{U_r(p,1)}{U_r(p,1-p)} = \frac{p}{2p(1-p)} = \frac{1}{2(1-p)}
\]
The result for false positive and negative rates follow immediately from the fact that they are $0$ for accurate and $1-q$ for noisy grades.
\end{proof}

\begin{lemma}\label{lem: effect_grades_fix_strat}
For an strategically signaling school, the impact on expected utility of moving between noisy and accurate grades is quantified by
\[1\leq \frac{U_s(p, 1)}{U_s(p, q)} \leq \frac{1}{1-p} < 2.\]
An strategically signaling school maximizes its expected utility when $q = 1$.

The impact on the FPR is
\[1\leq \frac{FPR_s(p, 1)}{FPR_s(p, q)} \leq \frac{1}{1-p} < 2.\]

The impact on the FNR, for $q \neq 1$, is
\[FNR_s(p, 1)=0,~FNR_s(p, q) = 1-q.\]

Further, all above bounds are tight for some $q$. 
\end{lemma}

\begin{proof}
The expected utility of a strategic school with noisy grades is
\begin{align*}
U_s(p,q) 
= 1 + (p + q - 2 p q) \cdot \frac{2 p - 1}{q - p}
\end{align*}
with $U_s(p,1) = 2p$. Because $U_s(p,q)$ is increasing in $q$ by Lemma~\ref{lem: monotonicity_results}, we have that $U_s(p,q) \leq U_s(p,1)$ and $\frac{U_s(p,1)}{U_s(p,q)} \geq 1$. Further, $q \geq 1-p$ implies
\[
U_s(p,q) \geq U_s(p,1-p) = 1 + (1 - 2p (1-p)) \cdot  \frac{2p - 1 }{1 - 2p} = 2p (1-p)
\]
Therefore, 
\[
\frac{U_s(p,1)}{U_s(p,q)} \leq \frac{1}{1-p} < 2,
\]
recalling that by Assumption~\ref{as: utility_noinfo}, $1 - p > 1/2$. The ratio of false negative rates is exactly $0$, as the false negative rate is $0$ when $q = 1$ and non-zero when $q \neq 1$. The false positive rate for accurate grades is $FPR_s(p,1) = \frac{p}{1-p}$. For noisy grades, $FPR_s(p,q)$ is increasing in $q$ by Lemma~\ref{lem: monotonicity_results}, and it must be the case that $FPR_s(p,q) \leq FPR_s(p,1)$. Further,
\[
FPR_s(p,q) \geq FPR_s(p,1-p) =  \frac{2p(1-p) - p}{1-p-p} = \frac{p - 2p^2}{1-2p} = p
\]
Hence, as $FPR_s(p,1) = \frac{p}{1-p}$ we have that 
\[
\frac{FPR_s(p, 1)}{FPR_s(p, q)} \leq \frac{1}{1-p} < 2.
\]
where the last inequality follows from $p < 1/2$.
\end{proof}

\begin{lemma}\label{lem: effect_strategic_fix_acc_g}
For a school with accurate grades, the impact on expected utility of introducing strategic signaling is 
\[\frac{U_s(p,1)}{U_r(p,1)} = 2.\]
The impact on the false positive rate is 
\[FPR_s(p,1) = \frac{p}{1-p},~FPR_r(p,1) = 0.\]
The impact on  the false negative rate is
\[FPR_s(p,1) = FNR_r(p,1) = 0.\]
The optimal signaling scheme doubles the expected utility of the school by increasing its false positive rate from $0$ to $\frac{p}{1-p}$, and keeping its false negative rate constant at $0$.
\end{lemma}

\begin{proof}
This follows immediately from $U_s(p,1) = 2p$, $FPR_s(p,1) = \frac{p}{1-p}$, $FNR_s(p,1) = 0$, $U_r(p,1) = p$, $FPR_r(p,1) = FNR(p,1) = 0$.
\end{proof}

\begin{lemma}\label{lem: effect_strategic_fix_noisy_g}
For a school with noisy grades, the impact on expected utility of introducing strategic signaling is
\[1 \leq \frac{U_s(p,q)}{U_r(p,q)} \leq 2,\]
and is increasing in $q$.

The impact on the false positive rate is
\[
1 \leq \frac{FPR_s(p,q)}{FPR_r(p,q)} \leq +\infty,
\]
and is increasing in $q \in [p-1,1]$. The impact on the false negative rate is 
\[
\frac{1-2p}{1-p}  \leq \frac{FNR_s(p,q)}{FNR_r(p,q)} \leq 1
\]
and is decreasing in $q \in [p-1,1]$. Further, all above bounds are tight for some $q$.
\end{lemma}

\begin{proof}
For a revealing school with noisy grades, the expected utility $U_r(p,q) = p q + (1-p)(1-q) = q (2p -1) + (1-p)$ and the false positive rate $FPR_r(p,q) = 1-q$ are decreasing in $q$ (as $2p-1 <0$ by Assumption \ref{as: utility_noinfo}). By Lemma~\ref{clm: min_q}, we have that $q \geq 1-p$ and it must be that 
\[
2p(1-p) = U_r(p,1-p) \geq U_r(p,q) \geq U_r(p,1) = p
\]
For an strategically signaling school, the expected utility $U_s(p,q)$ is increasing in $q$ by Lemma \ref{lem: monotonicity_results}, hence 
\[
 2p(1-p) = U_s(p,1-p) \leq U_s(p,q) \leq U_s(p,1) = 2p 
\]
The ratio of expected utilities is therefore increasing, and satisfies
\[
1 \leq \frac{U_s(p,q)}{U_r(p,q)} \leq 2.
\]
The false positive rate $FPR_s(p,q)$ is increasing in $q$ also by Lemma~\ref{lem: monotonicity_results}, hence $\frac{FPR_s(p,q)}{FPR_n(p,q)}$ is increasing in $q$. As $FPR_s(p,1-p) = p$, $FPR_s(p,q) = p$, $FPR_s(p,1) = \frac{p}{1-p}$ and $FPR_r(p,1) = 0$,
\[
1 = \frac{FPR_s(p,1-p)}{FPR_r(p,1-p)} \leq \frac{FPR_s(p,q)}{FPR_n(p,q)} \leq \frac{FPR_s(p,1)}{FPR_r(p,1)} = +\infty.
\]
$FNR_s(p,q) = (1-q) \frac{1-2p}{q-p}$ and $FNR_r(p,q) = 1-q$, hence the ratio of false negative rates for $q \neq 1$ is given by 
\[
H(p,q) = \frac{FNR_s(p,q)}{FNR_r(p,q)} = \frac{1-2p}{q-p} 
\]
which is a decreasing function of $q$, and we have 
\[
1 = H(p,1-p) \geq \frac{FNR_s(p,q)}{FNR_r(p,q)} \geq \lim_{q \to 1} H(p,q) = \frac{1-2p}{1-p}. \qedhere
\]
\end{proof}



\section{False positive and negative rates -- Model with standardized test}\label{app: false_rates_withexam}

\begin{lemma}[Revealing school's FPR/FNR]~\label{lem:revealing-fairness-exam}
In the presence of a standardized test, when a school is revealing, the false positive rate is given by 
\begin{align*}
FPR_r(p,q,\delta) =& (1-q) (1-\delta)  
+ \mathbbm{1} \left[ u_{q,\delta}(1,0) \geq 0 \right]   (1-q) \delta 
\\&+ \mathbbm{1} \left[ u_{q,\delta}(0,1) \geq 0 \right]  q (1-\delta) 
\end{align*}
and the false negative rate by
\begin{align*}
FNR_r(p,q,\delta)  =&  (1-q) (1-\delta) 
+ \mathbbm{1} \left[ u_{q,\delta}(1,0) < 0 \right] q (1-\delta)
\\&+ \mathbbm{1} \left[ u_{q,\delta}(0,1) < 0 \right] (1-q) \delta.
\end{align*}
For the special case of a revealing school with accurate grades (when $q = 1$), we have 
\[
FPR_r(p,1,\delta) = FNR_r(p,1,\delta) = 0.
\]
\end{lemma}

\begin{proof}
\begin{align*}
FPR_r(p,q) 
=& \Pr \left[g=1, s= 1 | t = 0 \right] 
\\&+ \mathbbm{1} \left[ u_{q,\delta}(1,0) \geq 0 \right] \Pr \left[g=1, s= 0 | t = 0 \right]
\\&+ \mathbbm{1} \left[ u_{q,\delta}(0,1) \geq 0 \right] \Pr \left[g=0, s= 1 | t = 0 \right]
\\=& (1-q) (1-\delta) 
+ \mathbbm{1} \left[ u_{q,\delta}(1,0) \geq 0 \right] (1-q) \delta 
\\&+ \mathbbm{1} \left[ u_{q,\delta}(0,1) \geq 0 \right]  q (1-\delta)
\end{align*}
and 
\begin{align*}
FNR_r(p,q) =& \Pr \left[g=0, s= 0 | t = 1 \right] 
\\&+ \mathbbm{1} \left[ u_{q,\delta}(1,0) < 0 \right] \Pr \left[g=1, s= 0 | t = 1 \right]
\\&+ \mathbbm{1} \left[ u_{q,\delta}(0,1) < 0 \right] \Pr \left[g=0, s= 1 | t = 1 \right]
\\=& (1-q) (1-\delta) 
+ \mathbbm{1} \left[ u_{q,\delta}(1,0) < 0 \right] q (1-\delta)
\\&+ \mathbbm{1} \left[ u_{q,\delta}(0,1) < 0 \right] (1-q) \delta.
\end{align*}
When $q = 1$, the university accepts a student if and only if his grade is $g = 1$, exactly all the high-type students get accepted, and it follows that $FPR(p,1,\delta) = 0 $, $FNR(p,1,\delta) = 0$.
\end{proof}

\begin{lemma}[Strategic school's FPR/FNR]\label{signaling-fairness-exam}
In the presence of a standardized test, when a school signals optimally and $u_{q,\delta}(1,0) < 0$, the false positive and negative rates are given by 
\begin{align*}
FPR_s(p,q,\delta) &= \Pr \left[s = 1 | t = 0 \right] = 1-\delta \\
FNR_s(p,q,\delta)  &= \Pr \left[s = 0 | t = 1 \right] = 1-\delta.
\end{align*}
When a school signals optimally and $u_{q,\delta}(1,0) \geq 0$, the false positive rate is given by 
\begin{align*}
FPR_s(p,q,\delta) 
&= (1 - q \delta) + q \delta \cdot \frac{p q (1-\delta) - (1-p) (1-q) \delta }{ (1-p) q \delta - p (1-q) (1-\delta)  }
\end{align*}
and the false negative rate by
\begin{align*}
FNR_s(p,q,\delta) 
&= (1-q)(1-\delta) \cdot \left( 1 - \frac{p q (1-\delta) - (1-p) (1-q) \delta }{ (1-p) q \delta - p (1-q) (1-\delta)  } \right)
\end{align*}

For the special case of a strategic school with accurate grades (when $q = 1$), we have 
\begin{align*}
FPR_s(p,1,\delta) = (1 - \delta) \frac{1}{1-p} 
\end{align*}
and 
\begin{align*}
FNR_s(p,1,\delta) = 0.
\end{align*}
\end{lemma}

\begin{proof}
When $u_{q,\delta}(1,0) < 0$, a student is accepted if and only if $s = 1$, and the false positive rate is given by
\begin{align*}
FPR_s(p,q,\delta) = \Pr \left[s = 1 | t = 0 \right] = 1-\delta
\end{align*}
and the false negative rate by
\begin{align*}
FNR_s(p,q,\delta) = \Pr \left[s = 0 | t = 1 \right] = 1-\delta.
\end{align*}
When $u_{q,\delta}(1,0) \geq 0$, a student gets rejected with probability $\Pr[\sigma^+ | g = 0, s = 0]$ when $s=g=0$. Therefore,
\begin{align*}
FPR_s(p,q,\delta) 
=& (1 - \Pr \left[g=0,s=0 | t=0 \right])\\
&+ \Pr \left[g=0,s=0 | t=0 \right]  \Pr[\sigma^+ | g = 0, s = 0]
\\=& (1 - q \delta) + q \delta \frac{p q (1-\delta) - (1-p) (1-q) \delta }{ (1-p) q \delta - p (1-q) (1-\delta)  }
\end{align*}
and 
\begin{align*}
FNR_s(p,q,\delta) 
&=\Pr \left[g=0,s=0 | t=1 \right] \Pr[\sigma^+ | g = 0, s = 0]
\\&= (1-q)(1-\delta) \left(1 - \frac{p q (1-\delta) - (1-p) (1-q) \delta }{ (1-p) q \delta - p (1-q) (1-\delta)  }\right).
\end{align*}

For the special case when $q = 1$, 
\begin{align*}
FPR_s(p,1,\delta) 
&= (1 - q \delta) + q \delta \frac{p q (1-\delta) - (1-p) (1-q) \delta }{ (1-p) q \delta - p (1-q) (1-\delta)  }
\\& = 1 - \delta + \delta  \frac{p (1-\delta) }{ (1-p) \delta }
\\& = (1 - \delta) + (1- \delta) \frac{p}{1-p} 
\\& = (1 - \delta) \frac{1}{1-p} 
\end{align*}
and 
\[
FNR_s(p,1,\delta) = 0. \qedhere
\]
\end{proof}

\section{Missing Lemmas and Proofs - Model With standardized Test}\label{app: proofs_withexam}

\begin{proof}[Proof of Lemma~\ref{lem: SAT+gradesconditions}]
By Observation~\ref{clm: min_q} and Assumption~\ref{as: utility_noinfo}, $q \geq 1-p > 1/2 > 1-q$. Therefore, $p q \delta - (1-p)(1-q) (1-\delta) \geq q \left( p \delta - (1-p) (1-\delta)  \right)$, which is non-negative by Assumption~\ref{as: util_SAT}, and $p (1-q) (1-\delta) - (1-p) q \delta < (1-q) (p (1-\delta) - (1-p) \delta$, which is negative by Assumption~\ref{as: util_SAT}. The rest of the proof follows from the fact that 
\begin{align*}
u_{q,\delta}(1,1) 
& = \frac{\Pr \left[t=1,g=1,s=1 \right]}{\Pr \left[g=1,s=1 \right]} - \frac{\Pr \left[t=0,g=1,s=1 \right]}{\Pr \left[g=1,s=1 \right]} 
\\& = \frac{p q \delta - (1-p)(1-q) (1-\delta)}{\Pr \left[g=1,s=1 \right]}
\end{align*}
and
\begin{align*}
u_{q,\delta}(0,0) 
& = \frac{\Pr \left[t=1,g=0,s=0 \right]}{\Pr \left[g=0,s=0 \right]} - \frac{\Pr \left[t=0,g=0,s=0 \right]}{\Pr \left[g=0,s=0 \right]} 
\\& = \frac{p (1-q) (1-\delta) - (1-p) q \delta}{\Pr \left[g=0,s=0 \right]} \qedhere
\end{align*}  
\end{proof}

\begin{proof}[Proof of Theorem~\ref{thm:opt-signaling-exam}]
The revelation principle of Lemma~\ref{clm: rev_principle} can be extended to the current setting via a nearly identical proof. Therefore, as before, we design $\sigma^+$ and $\sigma^-$ so that every student with signal $\sigma^+$ is accepted by the university, and every student with signal $\sigma^-$ is rejected. The school's goal is then to maximize the probability of a student having signal $\sigma^+$, under the constraint that the university gets expected non-negative expected utility from students with signal $\sigma^+$,
regardless of their score.

We first consider the case in which $s = 1$:
\begin{align*}
\Pr &\left[ t=1 | \sigma^+, s = 1 \right] 
= \frac{\Pr \left[\sigma^+, s=1 | t=1\right] \Pr \left[t = 1 \right]}{\Pr \left[\sigma^+, s= 1\right]}
\\&=p \cdot \frac{q \delta \Pr[\sigma^+ | g = 1, s = 1]  + (1-q) \delta \Pr[\sigma^+ | g = 0, s = 1]}{\Pr \left[\sigma^+, s= 1\right]}
\end{align*}
We also have that, by similar calculations: 
\begin{align*}
\Pr  \left[ t=0 | \sigma^+, s = 1 \right]  =& (1-p) \cdot \frac{(1-q) (1-\delta) \Pr[\sigma^+ | g = 1, s = 1]}{\Pr \left[\sigma^+, s= 1\right]}\\
&+ (1-p) \cdot\frac{q (1-\delta) \Pr[\sigma^+ | g = 0, s = 1]}{\Pr \left[\sigma^+, s= 1\right]}.
\end{align*}
Therefore, the university's expected utility for accepting a student with $(\sigma^+, s=1)$ is non-negative if and only if 
\begin{align*}
&p \left( q \delta \Pr[\sigma^+ | g = 1, s = 1] + (1-q) \delta \Pr[\sigma^+ | g = 0, s = 1]\right) \geq
\\&(1-p) \left((1-q) (1-\delta) \Pr[\sigma^+ | g = 1, s = 1] + q (1-\delta) \Pr[\sigma^+ | g = 0, s = 1] \right) ,
\end{align*}
which can be rewritten to give the constraint
\begin{align*}\label{eq: cond1+}
\Pr&[\sigma^+ | g = 1, s = 1] \left( p  q \delta -  (1-p) (1-q) (1-\delta) \right) \nonumber
\\& \geq  \Pr[\sigma^+ | g = 0, s = 1] \left(   (1-p) q (1-\delta) - p  (1-q) \delta \right). 
\end{align*}
By Assumption~\ref{as: util_SAT}, 
\begin{align*}
0 
&\leq p \delta - (1-p) (1-\delta) 
\\&= \left( p q \delta - (1-p) (1-q) (1-\delta)\right) 
\\&- \left( (1-p) q (1-\delta) -  p (1-q) \delta\right),
\end{align*}
and hence the constraint does not bind, and we are free to set
\[\Pr[\sigma^+ | g = 1, s = 1] = \Pr[\sigma^+ | g = 0, s = 1] = 1.
\]

We now consider the case in which the signal is $\sigma^+$ and the score is $s = 0$. Similar calculations to the $s = 1$ case show that the university's expected utility for accepting a student with such a score and signal is non-negative iff
\begin{align}\label{eq: cond0+}
\Pr&[\sigma^+ | g = 1, s = 0] \left(p  q (1-\delta) -  (1-p) (1-q) \delta \right) \nonumber
\\& \geq  \Pr[\sigma^+ | g = 0, s = 0] \left( (1-p) q \delta  - p  (1-q) (1-\delta) \right). 
\end{align}
Note that by Lemma~\ref{lem: SAT+gradesconditions},
$  (1-p) q \delta  - p  (1-q) (1-\delta) \geq 0.$

We split up the case in which $s = 0$ into two sub-cases on $u_{q,\delta}(1,0)$.

When $u_{q,\delta}(1,0) \geq 0$, then $p q (1-\delta) -  (1-p) (1-q) \delta \geq 0$.

Therefore, to maximize its expected utility, the school should set 
\begin{align*}
&\Pr[\sigma^+ | g = 1, s = 0] =1,
\\& \Pr[\sigma^+ | g = 0, s = 0] = \min \left(1, \frac{p  q (1-\delta) -  (1-p) (1-q) \delta }{ (1-p) q \delta -  p  (1-q) (1-\delta)  } \right) 
\end{align*}
Because by Assumption~\ref{as: util_SAT}, $p (1-\delta) - (1-p) \delta < 0$, it must be the case that 
\[
\frac{p  q (1-\delta) - (1-p) (1-q) \delta }{ (1-p) q \delta -  p  (1-q) (1-\delta)  } < 1, 
\]
and the school therefore optimizes its expected utility with
\[
\Pr[\sigma^+ | g = 0, s = 0] = \frac{p  q (1-\delta) -  (1-p) (1-q) \delta }{(1-p) q \delta - p (1-q) (1-\delta).  } 
\]

When $u_{q,\delta}(1,0) <0$, then  the left-hand side of Equation~\eqref{eq: cond0+} is non-positive. The right-hand side non-negative, so 
\[
\Pr[\sigma^+ | g = 1, s = 0] = \Pr[\sigma^+ | g = 0, s = 0] = 0
\]
is required for the inequality to hold. 
\end{proof}

\begin{proof}[Proof of Lemma~\ref{lem:revealing-util-exam}]
By Lemma~\ref{lem: SAT+gradesconditions}, $u_{q,\delta}(1,1) \geq 0$, and  hence the university accepts students with $g=1,s=1$, which occurs with probability $p q \delta + (1-p)(1-q)(1-\delta)$. 
Similarly,
$u_{q,\delta}(0,0) < 0$, so
the university rejects students with $g=0,~s=0$. 
The case $g=0, ~s=1$ happens with probability $p (1-q) \delta + (1-p) q (1-\delta)$, and yields the school expected utility $1$ iff $u_{q,\delta}(0,1) \geq 0$ (i.e., the university accepts students with $g=0,~s=1$). 
Similarly, $g=1,~s=0$ happens with probability $p q (1-\delta) + (1-p) (1-q) \delta$ and yields the school expected utility $1$ iff $u_{q,\delta}(1,0) \geq 0$. 


The second part of the proof has two cases.

When $q = 1$ and $\delta \neq 1$, then $u_{q,\delta}(1,0) = p q (1-\delta) - (1-p)(1-q) \delta = p (1-\delta) \geq 0$. Also,  $u_{q,\delta}(0,1) = p (1-q) \delta - (1-p) q (1-\delta) = - (1-p)(1-\delta) < 0$ (recalling that $1-p > 0$ by Assumption~\ref{as: utility_noinfo} and that $1-\delta > 0$). Therefore, for $\delta \neq 1$,
\begin{align*}
U_r(p,1,\delta) 
=& p q \delta + (1-p)(1-q)(1-\delta) + p q (1-\delta) + (1-p)(1-q) \delta 
\\=& p \delta + p (1-\delta)
\\=& p.
\end{align*}

When $q=1$ and $\delta = 1$, then $u_{q,\delta}(1,0) = p q (1-\delta) - (1-p)(1-q) \delta = p (1-\delta) \geq 0$. Also,  
$u_{q,\delta}(0,1)  = p (1-q) \delta - (1-p) q (1-\delta) = 0$. Therefore, 
\begin{align*}
U_r(p,1,1) 
=& p q \delta + (1-p)(1-q)(1-\delta)
\\&+ p q (1-\delta) + (1-p)(1-q) \delta  
\\&+  p (1-q) \delta + (1-p) q (1-\delta) \\
=& p .
\end{align*}
This concludes the proof. 
\end{proof}


\begin{proof}[Proof of Lemma~\ref{lem:signaling-util-exam}]
When $u_{q,\delta}(1,0) <0$, then
\begin{align*}
U_s(p,q,\delta) 
 =& \Pr \left[s = 1 \right] 
\\ =& p \delta + (1-p) (1-\delta). 
\end{align*}

When $u_{q,\delta}(1,0) \geq 0$, $\Pr[\sigma^+ | g = 1, s = 1] = \Pr[\sigma^+ | g = 1, s = 0] = \Pr[\sigma^+ | g = 0, s = 1] = 1$, and we have
\begin{align*}
U_s&(p,q,\delta) \\
=& 1 - \Pr \left[g=0, s=0 \right] +  \Pr \left[g=0, s=0 \right]  \Pr[\sigma^+ | g = 0, s = 0]
\\=& \left(1 - p (1-q) (1-\delta) - (1-p)q\delta \right)
\\&+ \left( p (1-q) (1-\delta) + (1-p) q \delta \right) \frac{p q (1-\delta) - (1-p) (1-q) \delta }{ (1-p) q \delta - p  (1-q) (1-\delta)} .
\end{align*}

When $q=1$, $u_{q,\delta}(1,0) \geq 0$ (equivalently, $p(1-\delta) \geq 0)$, and the expected utility is
\begin{align*}
U_s(p,1,\delta) &= \left(1 - (1-p)\delta \right)+ \left((1-p) \delta \right) \frac{p (1-\delta)}{ (1-p) \delta} \\
& = 1 - (1-p) \delta + p (1-\delta) 
\\&= 1 - \delta + p. \qedhere
\end{align*}
\end{proof}

\section{Relaxed assumption - Model with standardized test}\label{app: relaxed_assumption}


In this section, we relax Assumption~\ref{as: util_SAT} to instead  make the following, more general, assumption:
\begin{assumption}\label{as: relaxed_assumption}
$u_{q,\delta}(0,0) < 0$ and $u_{q,\delta}(1,1) \geq 0$.
\end{assumption}
When this assumption does not hold because $u_{q,\delta}(0,0) \geq 0$, the university has non-negative expected utility for every combination of test score and grade, and therefore accepts every student independently of how the school signals. When this assumption does not hold because $u_{q,\delta}(1,1) < 0$, the university has negative expected utility for every combination of test score and grade, and does not accept any student. We now present the optimal signaling scheme for the school under this relaxed assumption.

\begin{theorem}~\label{thm:opt-signaling-exam-relaxed}
The optimal signaling scheme for a school with access to grades and a test score, under Assumption~\ref{as: relaxed_assumption} is
\begin{align*}
\Pr \left[\sigma^+ |~g = 1,s = 1 \right] &= 1\\
\Pr \left[\sigma^+ |~g = 0,s = 1 \right] &= 
\begin{cases}
1, & \text{if}\ u_{q,\delta}(0,1) \geq 0\\
\min \left(1, \frac{p q \delta -  (1-p) (1-q) (1-\delta)}{ (1-p) q (1-\delta) - p (1-q) \delta} \right), & \text{if}\ u_{q,\delta}(0,1) < 0
\end{cases}\\
\Pr \left[\sigma^+ |~g = 1,s = 0 \right] &=
\begin{cases}
1, & \text{if}\ u_{q,\delta}(1,0) \geq 0 \\ 
0, & \text{if}\ u_{q,\delta}(1,0) < 0
\end{cases}\\
\Pr \left[\sigma^+ |~g = 0,s = 0 \right] &= 
\begin{cases}
\min \left(1,\frac{p q (1-\delta) - (1-p) (1-q) \delta }{ (1-p) q \delta - p (1-q) (1-\delta)  }\right), & \text{if}\ u_{q,\delta}(1,0) \geq 0\\ 
0, & \text{if}\ u_{q,\delta}(1,0) < 0
\end{cases}
\end{align*}
\end{theorem}

\begin{proof}
The proof follows the same steps as that of Theorem \ref{thm:opt-signaling-exam}, minus the simplification steps that rely on Assumption \ref{as: util_SAT}.
\end{proof}

\section{Impact of standardized test}\label{app: ratio_utilities_with_exam}
\begin{lemma}
Fix $p > 0$ and $q < 1$. For 
\[
\delta > \max \left( \frac{pq}{(pq + (1-p)(1-q)},\frac{(1-p)q}{p(1-q) + (1-p) q}\right),
\]
we have
\[
\frac{U_s(p,q,\delta)}{U_r(p,q,\delta)} = 1.
\]
\end{lemma}

\begin{proof}
$p q (1-\delta) - (1-p) (1-q) \delta < 0$ and $p (1-q) \delta - (1-p) q (1-\delta) \geq 0$, hence $u_{q,\delta}(1,0) < 0$ and $u_{q,\delta}(0,1) \geq 0$. Therefore, a school's expected utility for revealing is given by
\[
p q \delta + (1-p)(1-q) (1-\delta) 
+ p (1-q) \delta + (1-p) q (1-\delta) 
= p \delta + (1-p) (1-\delta)
\]
This is exactly the expected utility of a school that is strategically signaling when $q = 1$, hence the ratio of utilities when strategic over revealing is $1$.
\end{proof}

\begin{lemma}
When the grades are accurate, i.e. $q=1$, 
\[1 \leq \frac{U_s(p,1,\delta)}{U_r(p,1,\delta)} \leq 2,\]
\end{lemma}

\begin{proof}
The expected utility from strategically reporting is $p + (1-p) (1-\delta) + p (1-\delta)$ by Lemma~\ref{lem:signaling-util-exam}, while the expected utility for revealing is $p$ by Lemma~\ref{lem:revealing-util-exam}. Therefore, 
\[
\frac{U_s(p,1,\delta)}{U_r(p,1,\delta)}   = 2-\delta + \frac{1-p}{p} (1-\delta) \geq 1 .
\]
By Assumption~\ref{as: grades}, $ (1-p) (1-\delta) \leq p \delta$, hence  
\[
\frac{U_s(p,1,\delta)}{U_r(p,1,\delta)}   = 2-\delta + \frac{1-p}{p} (1-\delta) \leq 2 - \delta + \delta = 2.\qedhere
\]
\end{proof}
We have by Theorem~\ref{thm: utility_comparison} that in the absence of standardized test, 
\[
\frac{U_s(p,1)}{U_r(p,1)}  = 2.
\]
Forcing students to take a standardized test thus improves fairness between two schools with accurate grades.

\end{document}